\documentclass[submission,copyright,creativecommons,fleqn]{eptcs}
\providecommand{\event}{AFL 2023} 

\usepackage{amssymb,amsmath,latexsym}
\usepackage[utf8]{inputenc}
\usepackage{tikz}
\usetikzlibrary{matrix,arrows,automata,positioning}
\usepackage{stmaryrd}
\usepackage{xcolor}
\usepackage{mathrsfs}
\usepackage{graphicx}
\newcommand{\beh}[1]{\llbracket #1\rrbracket}
\newcommand{\overbar}[1]{\mkern 1.7mu\overline{\mkern-1.7mu#1\mkern-1.7mu}\mkern 1.7mu}

\newtheorem{theorem}{Theorem}[section]

\newtheorem{example}[theorem]{Example}

\def\endproof{\qed\endtrivlist}
\expandafter\let\csname endproof*\endcsname=\endproof

\def\qedsymbol{\ifmmode\bgroup\else$\bgroup\aftergroup$\fi
  \vcenter{\hrule\hbox{\vrule height.6em\kern.6em\vrule}\hrule}\egroup}
\def\qed{\ifmmode\else\unskip\nobreak\fi\quad\qedsymbol}

\usepackage{iftex}

\ifpdf
  \usepackage{underscore}         
  \usepackage[T1]{fontenc}        
\else
  \usepackage{breakurl}           
\fi

\DeclareFontFamily{OT1}{pzc}{}
\DeclareFontShape{OT1}{pzc}{m}{it}{<-> s * [1.200] pzcmi7t}{}
\DeclareMathAlphabet{\mathpzc}{OT1}{pzc}{m}{it}

\title{Weighted Automata over Vector Spaces\footnote{This research was supported by the Science Fund of the Republic of Serbia, Grant no 7750185, Quantitative Automata Models: Fundamental Problems and Applications - QUAM}}
\author{Nada Damljanovi\'c
\institute{University of Kragujevac, Faculty of Technical Sciences, Svetog Save 65, \v Ca\v cak, Serbia}
\email{nada.damljanovic@ftn.kg.ac.rs}
\and
Miroslav \'Ciri\'c\qquad\qquad Jelena Ignjatovi\'c
\institute{University of Ni\v s, Faculty of Sciences and Mathematics, Vi\v segradska 33, Ni\v s, Serbia}
\email{\quad miroslav.ciric@pmf.edu.rs \quad\qquad jelena.ignjatovic@pmf.edu.rs}
}
\def\titlerunning{Weighted Automata over Vector Spaces}
\def\authorrunning{N. Damljanovi\'c, M. \'Ciri\'c \& J. Ignjatovi\'c}
\begin{document}
\maketitle

\begin{abstract}
In this paper we deal with three models of weighted automata that take  weights in the field of real numbers.~The first of these models are classical weighted finite automata, the second one are crisp-deter\-ministic weighted automata, and the third one are weighted automata over a vector space.~We explore the interrelationships between weighted automata over a vector space and other two models. 
\end{abstract}

\section{Introduction}

Weighted automata belong to the fundamental models of computation in computer science.~They can be understood as an extension of conventional automata in which the transitions and states carry numerical or other values called weights.~These weights may model quantitative proper\-ties like the cost, the amount of resources needed for the execution of a transition, the reliability or probability of the successful execution of the transitions, or many other things.~Different models of weighted automata differ in the algebraic structures within which the weights are taken, as well as in the way in which these weights are manipulated.

In this paper, we deal with weighted automata that take weights in the field of real numbers.~Such automata have been the subject of study since the very beginning of the theory of weighted automata, since the seminal work of Sch\"utzenberger \cite{Schutz.61} who studied weighted automata over the field.~Today, they are very popular due to their significant applications, primarily in formal specification and verification of systems, as well as in the field of machine learning, where they are successfully used as an alternative to recurrent neural networks.~We discuss three models of weighted automata with weights taken in the field of real numbers.

The first of these models are classical \textit{weighted finite automata}.~The common way of viewing deterministic and nondeterministic finite automata as labelled graphs has also been used for weighted finite automata from the very beginning of their studying.~From such a point of view, a weighted finite automaton is represented by a directed multi-graph whose edges carry two labells, the input~letter and the weight, while nodes carry two weights, the initial and terminal weight.~The computation along a path in the graph is performed by concatenation of the input letters and multiplication of the initial weight of the starting node, the weights of edges along the path, and the terminal weight of the final node, and then the sum of the weights of all paths labelled with the same input word is computed and assigned to this input word.~This determines the behavior of the considered weighted finite automaton, that is, the word function computed by that automaton.~Such an understanding of the behaviour can be called the \textit{dept-first semantics}.~Another way of looking at weighted finite automata, through vector and matrix operations, has also been present since their very beginnings.~From that point of view, the behaviour of a weighted finite automaton can be expressed as the product of the row vector representing the initial weights, matrices representing the weights of the transitions induced by input letters, and the column vector representing the terminal weights.~Such a representation of a weighted finite automaton is called a \textit{linear representation}, while such an understanding of the behaviour can be called the \textit{breadth-first semantics}.~In the case of weighted finite automata over a semiring these two semantics are the same.~The linear algebraic approach proved to be extremely powerful and useful, especially in the study of simulations and bisimulations, as well~as in the reduction of the number of states.~That approach was successfully applied to nondeterministic finite automata \cite{CIBJ.14}, fuzzy finite automata \cite{CIDB.12,CIJD.12,CSIP.10,SCI.14,SCB.18} and weighted finite automata over an additively-idempotent semiring \cite{DCI.14}, and research is underway in which that approach is applied in the context of weighted automata over the max-plus semiring and the field of real numbers.~The mentioned approach also plays a key role in this paper. 

The second model of weighted automata that we deal with here are the so-called \textit{crisp-determin\-is\-tic weighted automata}.~These are classical automata with a single initial state and deterministic transitions in which the set of terminal states is replaced by a function which assigns a terminal weight to each state.~When such an automaton starts working from the initial state and performs a sequence of transitions conducted by a given input word, the weight assigned to that word is the terminal weight of the destination state.~Those automata were studied for the first time in \cite{ICB.08}, in~the context of fuzzy automata, and the most general definition of crisp-deterministic weighted automata~was~given~in~\cite{ICBP.10}. The name crisp-deterministic was introduced in \cite{CDIV.10} to distinguish it from a related type of automata for which the name deterministic weighted automata is used.~An extensive study of crisp-determinis\-tic weighted automata was carried out in \cite{ICBP.10}, and in \cite{CDIV.10,ICB.08,JC.14,JIC.11,JMIC.16,MJIC.15} various procedures for~converting a weighted finite automaton into an equivalent crisp-deterministic weighted automaton were provided.~Such procedures are called \textit{crisp-determinization}.~If we allow a crisp-deterministic weighted automaton to have an infinite set of states, as we do in this paper, then any weighted finite automaton can be converted into an equivalent crisp-de\-ter\-ministic weighted automaton, and the basic problem is to perform such a conversion that will provide an equivalent crisp-deterministic weighted automaton with a finite number of states, as small as possible.~For information on crisp-de\-ter\-minization of weighted tree automata we refer to \cite{DFKV.20,FKV.21}.

The main role in the crisp-determinization is played by the concept of the \textit{Nerode automaton} assigned to the weighted finite automaton that is determinized.~The construction of the Nerode autom\-aton was first introduced in \cite{ICB.08} as a counterpart to the \textit{accessible subset construction} on which the determinization of classical nondeterministic finite automata is based.~According to that construction, the states of a Nerode automaton are vectors with entries from the underlying structure of weights, but in the mentioned papers dealing with crisp-determinization, such nature of states was neglected, and the Nerode automaton was considered as an ordinary crisp-deterministic weighted automaton.~If the vector nature of states is taken into account, this leads us to the third model of weighted automata that is considered here, to \textit{weighted automata over a vector space} or \textit{weighted automata with vector states}.~Various forms of such automata were studied in \cite{BGP.17,BGP.22,BBBRS.12,Bor.09,CP.17,LRP.18}, and a related model of au\-tom\-ata, called \textit{automata with fuzzy states}, was studied within the framework of fuzzy automata theory (see \cite{SCI.15} and sources cited there).~The concept of a weighted automaton over a vector space discussed here differs slightly from the corresponding concepts studied in the cited articles.~The first difference concerns the underlying vector space.~Except in \cite{CP.17}, in all the other mentioned articles, it is assumed that this vector space is finite-dimensional.~Here we not only allow that space to be infinite-dimensional, but also introduce an extremely interesting weighted automaton over an infinite-dimensional space, the so-called \textit{derivative automaton}.~The second difference concerns the set of states of these automata.~In all the mentioned articles, except in \cite{BBBRS.12}, states are assumed to be all vectors from the underlying vector space $V$.~However, in that case the set of states is always infinite and a huge number of states are unreachable from the initial state, and therefore redundant.~For this reason, we take the set of states to be a subset of $V$, which can be both finite and infinite.~The third difference relates to transition functions.~In almost all cited articles, the transition functions induced by the input letters were required to be linear operators on $V$.~In \cite{LRP.18}, a more general model of weighted automata over a vector space was proposed, where the transition functions~do not have to be linear.~This leads to the distinction between \textit{linear} and \textit{nonlinear} weighted automata over a vector space.~Here we give a  definition of a linear weighted automaton over a vector space which also includes the cases when the set of states is not the entire vector space and when the underlying vector space is infinite-dimensional.

This paper is the beginning of our extensive investigations of weighted automata with weights taken in the field of real numbers, and our aim here is to examine some general relations between weighted automata over vector spaces and other two models.~First, by Theorem 4.1, we show that any crisp-de\-ter\-ministic weighted automaton can be naturally turned into a language-equivalent weighted autom\-aton over a vector space, where the set of vector states can be any set of vectors that has the same cardinality as the set of states of that crisp-deterministic weighted automaton.~Then by Theorem 4.2 we show that any finite-dimensional linear weighted automaton over a vector space can be turned into a completely language-equivalent weighted finite automaton, and conversely, any weighted finite automaton can be turned into a completely language-equivalent finite-dimensional linear weighted automaton over a vector space.~Actually, we show that the previously mentioned Nerode automaton of a weighted finite automaton $\mathpzc{A}$ is a finite-dimensional linear weighted automaton over a vector space that is completely equivalent to $\mathpzc{A}$.~Theorem 4.3 gives us an elegant procedure for checking whether a given finite weighted automaton over a vector space is linear.~At the end of the paper, we introduce the concept of the derivative automaton of a given word function and prove that it is a linear weighted automaton over a vector space that computes this word function and generates its prefix closure.~In addition, we show that the derivative automaton is a minimal weighted automaton over a vector space which computes this word function.

\section{Preliminaries}

Throughout this paper, $\mathbb N$ denotes the set of all natural numbers (without zero)  and $\mathbb R$ denotes the field of real numbers. For $i, j \in \mathbb N$ such that $i \leq j$  we use the notation $[i..j] = \{k \in \mathbb{N}\, |\, i \leq k \leq j\}$.

A \textit{vector space} over $\mathbb{R}$ is a triple $(V,+,\cdot)$ such that:
\begin{itemize}
    \item[$\ast$] $V$ is a non-empty set, whose members are called \textit{vectors}; 
    \item[$\ast$] $+:V\times V\to V$ given by $+:(\alpha,\beta)\mapsto \alpha+\beta $, for $\alpha,\beta\in V$, is the \textit{vector addition} operation;
    \item[$\ast$] $\cdot :\mathbb{R}\times V\to V$ given by $\cdot:(r,\alpha)\mapsto r\cdot \alpha $, for $r\in \mathbb{R}$, $\alpha\in V$, is the \textit{scalar multiplication} operation;
    \item[$\ast$] vector addition and scalar multiplication satisfy the following axioms:
    \begin{itemize}
        \item[(V1)] $(V,+)$ is a commutative group,
        \item[(V2)] $r\cdot (\alpha+\beta)=r\cdot \alpha+r\cdot \beta$,
        \item[(V3)] $(r+s)\cdot \alpha=r\cdot \alpha+s\cdot \alpha$,
        \item[(V4)] $(r\cdot s)\cdot \alpha=r\cdot (s\cdot \alpha)$,
        \item[(V5)] $1\cdot \alpha =\alpha $,
    \end{itemize}
    for all $r,s\in \mathbb{R}$ and $\alpha,\beta\in V$.
\end{itemize} 
Note that a vector space over an arbitrary field can be defined in the same way.

The basic example of a vector space over $\mathbb{R}$ is the vector space $\mathbb{R}^n$ consisting of all $n$-tuples of real numbers, with vector addition and scalar multiplication defined coordinatewise.~Another example of~a~vector space over $\mathbb{R}$ that is important here is the vector space $\mathbb{R}^T$ consisting of all functions from a set $T$~into~$\mathbb{R}$, with vector addition and scalar multiplication defined by $(\alpha+\beta)(t)=\alpha(t)+\beta(t)$ and $(r\cdot \alpha)(t)=r\cdot \alpha(t)$, for all $\alpha,\beta\in \mathbb{R}^T$, $r\in \mathbb{R}$ and $t\in T$.~Such vector spaces are called \textit{function spaces}.

Let $V$ and $W$ be vector spaces over $\mathbb{R}$.~A function $h:V\to W$ is called a \textit{homomorphism} or \textit{linear transformation} of $V$ into $W$ if $h(\alpha+\beta)=h(\alpha)+h(\beta)$ and $h(r\cdot \alpha)=r\cdot h(\alpha)$, for all $\alpha,\beta\in V$ and $r\in \mathbb{R}$.~If~$h$ is a bijective homomorphism, then it is called an \textit{isomorphism} of $V$ into $W$, and we say that~$V$~and $W$ are \textit{isomorphic} vector spaces.~A vector space $V$ over $\mathbb{R}$ is said to be \textit{finite-dimensional} if it is isomorphic to the vector space $\mathbb{R}^n$, for some $n\in \mathbb{N}$.~In this case $n$ is the unique natural number having this property and it is called the \textit{dimension} of $V$.~A vector space which is not finite-dimensional is called \textit{inifinite-dimensional}.~A homomorphism (linear transformation) of a vector space $V$ into itself is called a \textit{linear operator} on $V$.

Let $V$ be a vector space over $\mathbb{R}$.~A \textit{linear combination} of vectors $\alpha_1,\alpha_2,\ldots,\alpha_k\in V$ is any expression of the form $r_1\cdot\alpha_1+r_2\cdot\alpha_2+\cdots +r_k\cdot\alpha_k$, where $r_1,r_2,\ldots,r_k\in \mathbb{R}$.~For any set $S\subseteq V$, the set of all linear combinations of vectors from $S$ is called the \textit{span} of $S$ and denoted by $\mathrm{span}(S)$.~In other words,
\[
\mathrm{span}(S)=\{\alpha\in V\,|\, (\exists k\in \mathbb{N})(\exists \alpha_1,\alpha_2,\ldots,\alpha_k\in S)(\exists r_1,r_2,\ldots,r_k\in \mathbb{R})\, \alpha=r_1\cdot\alpha_1+r_2\cdot\alpha_2+\cdots +r_k\cdot\alpha_k\}.
\]
It is well-known that $\mathrm{span}(S)$ is a vector space with vector addition and scalar multiplication inherited from $V$, i.e., it is a \textit{subspace} of $V$.

Given natural numbers $m, n \in \mathbb{N}$. A \textit{matrix} of type $m \times n$ with entries in the field of real numbers $\mathbb R$, or
a real $m \times n$-matrix, is  defined as a mapping $M : [1..m] \times [1..n] \to \mathbb R$. For a pair $(i, j) \in [1..m] \times [1..n]$ the value $M(i, j)$ is called the $(i, j)$-entry of the matrix $M$.~The set of all real matrices of type $m\times n$ is denoted by ${\mathbb R}^{m\times n}$. Similarly, a \textit{vector} of length $n$ with entries in $\mathbb R$, or real vector is defined as a mapping $\nu : [1..n] \to {\mathbb R}$. For each $i \in [1..n]$ the value $\nu(i)$ is called the $i$th entry or $i$th coordinate of the vector $\nu$. The set of all real vectors of length $n$  is denoted by $\mathbb{R}^n$.

The zero matrix of type $m \times n$, denoted by $O_{m\times n}$, is a matrix of type $m \times n$ whose all entries are~$0$. Similarly, the zero vector of length $n$, denoted by $o_n$, is a vector of length $n$ whose all entries are $0$. For each $n \in \mathbb{N}
$, a matrix of type $n \times n$ is called a \textit{square matrix} of order $n$.~The identity matrix of order $n$, denoted by $I_n$, is a square matrix of order $n$ which satisfies $I_n(i, i) = 1$, for each $i \in [1..n]$, and $I_n(i, j) = 0$,
for all $i, j \in [1..n]$ such that $i \neq j$.~The transpose of a matrix $M$ is denoted by $M^{\top}$.~For a matrix $M\in \mathbb{R}^{m\times n}$ and $k\in [1..n]$, by $c_k(M)$ we denote the $k$th column vector of $M$.

For all pairs of matrices from ${\mathbb R}^{m\times n}$ the \textit{matrix addition} is defined pointwise: 
\begin{equation}
(M + N)(i, j) = M(i, j) + N(i, j),
\end{equation}
for all $M,N \in {\mathbb R}^{m\times n}, i \in [1..m]$ and $j \in
[1..n]$. It is an associative and commutative operation on ${\mathbb R}^{m\times n}$, and in particular, $({\mathbb R}^{m\times n},+,O_{m\times n})$
forms a commutative monoid. The \textit{matrix product} is defined between matrices from ${\mathbb R}^{m\times n}$ and ${\mathbb R}^{n\times p}$
as follows: for $M \in {\mathbb R}^{m\times n}$ and $N \in {\mathbb R}^{n\times p}$ their product is a matrix $M \cdot N \in  {\mathbb R}^{m\times p}$ with entries
given by
\begin{equation}
(M \cdot N)(i, k) = \sum_{j=1}^n
M(i, j) \cdot M(j, k), 
\end{equation}
for all $(i, j) \in [1..m]\times[1..p]$. The matrix product is associative whenever it is defined, i.e.,
$(M \cdot N) \cdot P = M \cdot (N \cdot P)$, for all $M \in {\mathbb R}^{m\times n}$, $N \in {\mathbb R}^{n\times p}$ and $P \in  {\mathbb R}^{p\times q}$. In
particular, $({\mathbb R}^{n\times n},+, \cdot,O_{n\times n}, I_n)$ is a semiring. Given a matrix $M \in  {\mathbb R}^{m\times n}$ and vectors $\mu \in {\mathbb R}^m$
and $\nu \in  {\mathbb R}^n$. When $\mu$ is treated as a matrix of type $1\times m$ (row vector) and $\nu$ as a matrix of type $n\times 1$ (column vector), the \textit{vector-matrix product} $\mu \cdot M$ and the \textit{matrix-vector product} $M \cdot \nu$ are defined as
matrix products. The \textit{scalar product} or \textit{dot product} of vectors $\mu, \nu \in {\mathbb R}^n$ is an element of $\mathbb R$ given by
\begin{equation}
\mu \cdot \nu =\sum_{i=1}^n
\mu(i) \cdot \nu(i). 
\end{equation}

A matrix $M\in \mathbb{R}^{m\times n}$ is said to be in the \textit{row echelon form} if it satisfies the following properties:
\begin{itemize}
    \item[$\ast$] If a row of $M$ does not consist entirely of zeros, then the first nonzero entry in this row is $1$. It is called a \textit{leading} $1$.
    \item[$\ast$] If there are any rows that consist entirely of zeros, then they are grouped together at the bottom of the matrix $M$.
    \item[$\ast$] In any two successive rows of $M$ that do not consist entirely of zeros, the leading $1$ in the~lower row occurs farther to the right than the leading $1$ in the higher row.
\end{itemize}
Moreover, $M$ is said to be in the \textit{reduced row echelon form} if, in addition to these three properties, it also satisfies the condition
\begin{itemize}
    \item[$\ast$] Every column of $M$ that contains a leading $1$ has zeros everywhere else.
\end{itemize}
Every matrix $N\in \mathbb{R}^{m\times n}$ can be transformed to a row echelon form or a reduced row echelon form by applying some sequence of \textit{elementary row operations} (multiplying a row by a nonzero scalar, interchanging two rows, and adding a multiple of one row to another).~It should be noted that the reduced row echelon form of the matrix $N$ is unique, in the sense that reducing the matrix $N$ to the reduced row echelon form by applying any sequence of elementary row operations always yields the same matrix in the reduced row echelon form. This matrix will be denoted by $RREF(N)$.~The \textit{rank} of a matrix $N$, denoted by $\mathrm{rank}(N)$, is defined as the number of nonzero rows in $RREF(N)$. 

For matrices $M_1\in \mathbb{R}^{m\times n_1}$, $M_2\in \mathbb{R}^{m\times n_2}$, \ldots , $M_k\in \mathbb{R}^{m\times n_k}$, where $k,m,n_1,n_2,\ldots ,n_k\in \mathbb{N}$, by concatenating them from left to right we obtain a matrix $[\,M_1\mid M_2\mid\ldots \mid M_k\,]\in \mathbb{R}^{m\times n}$, where $n=n_1+n_2+\cdots+n_k$, which is called the \textit{augmented matrix} (obtained from $M_1$, $M_2$,\ldots , $M_k$).

For undefined notions and notation concerning vector spaces, vectors and matrices we refer to the book \cite{AB.23}.

\section{Three models of weighted automata}

In terms of real matrices and their properties, we will investigate three models of weighted automata with weights in the field of real numbers.

\subsection{Weighted finite automata}

Let $X$ be an alphabet.~A \textit{weighted finite automaton} over $\mathbb{R}$ and $X$ is defined~as~a tuple $\mathpzc{A}=(A,\sigma,\delta,\tau)$, where $A$ is a non-empty finite set, while $\sigma,\tau:A\to \mathbb{R}$ and $\delta:A\times X\times A\to \mathbb{R}$.~The function $\delta $ is often replaced by the family of functions $\{\delta_x\}_{x\in X}$, where $\delta_x:A\times A\to \mathbb{R}$ is given by $\delta_x(a,b)=\delta(a,x,b)$, for all $a,b\in A$ and $x\in X $. We call $A$ the \textit{set of states}, $\sigma $ the \textit{initial weights function}, $\tau $ the \textit{terminal weights function}, and $\delta $ and $\delta_x$, $x\in X$, the \textit{transition weights functions}.

The behavior of the weighted finite automaton $\mathpzc{A}$ is a word function $\beh{\mathpzc{A}}:X^*\to \mathbb{R}$ defined by
\begin{equation}\label{eq:behA}
\beh{\mathpzc{A}}(u)=\sum_{(a_0,a_1,\ldots,a_k)\in A^{k+1}} \sigma(a_0)\cdot \delta(a_0,x_1,a_1)\cdot \delta(a_1,x_2,a_2) \cdots \delta(a_{k-1},x_k,a_k)\cdot \tau (a_k), 
\end{equation}
for $u=x_1x_2\dots x_k\in X^+$, $x_1,x_2,\ldots ,x_k\in X$, and 
\begin{equation}\label{eq:behAe}
\beh{\mathpzc{A}}(\varepsilon)=\sum_{a\in A} \sigma(a)\cdot \tau (a). 
\end{equation}
We say that $\mathpzc{A}$ \textit{computes} the function $\beh{\mathpzc{A}}$.

Assume that $n$ is the number of elements of $A$, i.e., $A=\{a_1,a_2,\ldots,a_n\}$. In many situations, instead of as functions, it is more convenient to treat $\sigma $ and $\tau $ as vectors in $\mathbb{R}^n$, and $\delta_x$, $x\in X$, as $n\times n$ matrices with entries in $\mathbb{R}$. In other words, $\sigma $ will be identified with a vector in $\mathbb{R}^n$ whose $i$th coordinate~is~$\sigma(a_i)$, and $\tau $ will be identified with a vector in $\mathbb{R}^n$ whose $i$th coordinate is $\tau(a_i)$. In order to clearly distinguish between matrices and vectors, we will use capital letters of the Latin alphabet to denote matrices, while vectors will be denoted by small letters of the Greek alphabet.~Therefore, $\{M_x\}_{x\in X}$ will be a family of $n\times n$ matrices over $\mathbb{R}$ such that the $(i,j)$-entry of $M_x$~is equal to $\delta_x(a_i,a_j)$.~A weighted finite automaton $\mathpzc{A}$ is then treated as a tuple $\mathpzc{A}=(n,\sigma,\{M_x\}_{x\in X},\tau)$, where we cal $n$ the \textit{dimension}, $\sigma$ the \textit{initial weights vector}, $\tau$ the \textit{terminal weights vector}, and $M_x$, $x\in X$, the \textit{transition weights matrices}~of~$\mathpzc{A}$. We cal this tuple the \textit{linear representation} of the weighted finite automaton $\mathpzc{A}$.

When $\mathpzc{A}$ is given by the linear representation, its behavior is represented by
\begin{equation}\label{eq:behAl}
\beh{\mathpzc{A}}(u)=\sigma \cdot M_{x_1}\cdot M_{x_2}\cdots M_{x_k}\cdot \tau = \sigma\cdot M_u\cdot \tau , 
\end{equation}
for $u=x_1x_2\dots x_k\in X^+$, $x_1,x_2,\ldots ,x_k\in X$, where $M_u=M_{x_1}\cdot M_{x_2}\cdots M_{x_k}$, and 
\begin{equation}\label{eq:behAel}
\beh{\mathpzc{A}}(\varepsilon)=\sigma \cdot \tau . 
\end{equation}

In applications of automata in the theory of discrete event systems, apart from the function $\beh{\mathpzc{A}}$ computed by the automaton $\mathpzc{A}$, another function plays an important role -- the function $\beh{\mathpzc{A}}_g$ generated by the automaton $\mathpzc{A}$. For a weighted finite automaton  $\mathpzc{A}=(n,\sigma,\{M_x\}_{x\in X},\tau)$ this function can be defined with:
\begin{equation}\label{eq:genA}
    \beh{\mathpzc{A}}_g (u)=\Vert\sigma_u\Vert_{\infty} ,
\end{equation}
for every $u\in X^*$, where $\sigma_u=\sigma\cdot M_u$, for $u\in X^+$, and $\sigma_\varepsilon =\sigma $, while $\Vert\cdot\Vert_{\infty}$ denotes the \textit{maximum norm} (called also a \textit{uniform norm}) on $\mathbb{R}^n$ that is given by
\begin{equation}\label{eq:max.norm}
\Vert \alpha\Vert_{\infty}=\max_{i\in [1..n]} |\alpha_i|,    
\end{equation}
for each $\alpha=(\alpha_1,\alpha_2,\ldots ,\alpha_n)\in \mathbb{R}^n$.

Let us note that in the case of convential nondeterministic finite autom\-ata, $\beh{\mathpzc{A}}_g$ is the language consisting of all words for which a transition is defined, i.e., of all words which "lead somewhere"(cf.~\cite{CL.08}).~A nondeterministic finite automaton can be considered as a weighted finite automaton over the two-element Boolean semiring, and then $\beh{\mathpzc{A}}_g$ consists of all words $u$ for which $\sigma_u$ is a non-zero Boolean vector, i.e., for which 
\[
\max_{i\in [1..n]} |\alpha_i|=1,
\]
where $\sigma_u=(\alpha_1,\alpha_2,\ldots ,\alpha_n)\in \{0,1\}^n$. From conventional nondeterministic finite automata, such a definition was also extended to the case of fuzzy finite automata, where $\beh{\mathpzc{A}}_g(u)$ is interpreted as the degree to which the word $u$ leads somewhere.~A similar interpretation can be given here as well, when it comes to weighted finite automata over the field of real numbers, where $\beh{\mathpzc{A}}_g(u)=\Vert\sigma_u\Vert_{\infty}$~could~be interpreted as the maximal probability of the existence of a transition determined by $u$, i.e. the probability that $u$ leads somewhere.~In order for this to be consistent with the conventional view of probability, the values for $\Vert\sigma_u\Vert_{\infty}$, which are always non-negative, can be translated by some monotone function (for example, by the function $1-e^{-x})$, to values from the real unit interval $[0,1]$.

\subsection{Crisp-deterministic weighted automata}

Another model of weighted automata is a \textit{crisp-deterministic weighted automaton}, which is defined as a tuple $\mathpzc{D}=(D,d_0,\Delta ,\theta )$, where $D$ is a non-empty \textit{set of states}, $d_0\in D$ is a distinguished state that~is called the \textit{initial state}, $\Delta :D\times X\to D$ is a function called the \textit{transition function}, and $\theta :D\to \mathbb{R}$ is a function called the \textit{terminal weights function}, or the \textit{terminal weights vector}, if $\theta $ is considered as a vector in the space $\mathbb{R}^D$. Here it is not necessary that the set $D$ is finite, so we will also allow~the~possibility that $D$ is infinite.~If the set of states $D$ is finite, then $\mathpzc{D}$ is called a 
\textit{finite crisp-deterministic weighted automaton}

A finite crisp-deterministic weighted automaton can be considered as a special weighted finite automaton $\mathpzc{A}=(D,\sigma,\delta,\tau)$ in which for every $a\in D$ and $x\in X$ there exists $a'=\Delta(a,x)\in D$ such that $\delta(a,x,a')=1$, while $\delta(a,x,b)=0$, for every $b\in D\setminus\{a'\}$, and there exists $a\in D$ such that $\sigma(a)=1$, while $\sigma(b)=0$, for every $b\in D\setminus \{a\}$ (then we assume that $d_0=a$). In other words, a finite crisp-deterministic weighted automaton is a weighted finite automaton with a single initial state and a deterministic transition function, in which all weights are concentrated in the terminal weights vector.

The transition function $\Delta$ of a crisp-deterministic weighted automaton $\mathpzc{D}=(D,d_0,\Delta ,\theta )$ extends to a function $\Delta^*:D\times X^*\to D$ by putting $\Delta^*(a,\varepsilon)=a$ and $\Delta^*(a,ux)=\Delta (\Delta^*(a,u),x)$, for all $a\in D$,~$u\in X^*$ and $x\in X$, and the behavior $\beh{\mathpzc{D}}:X^*\to \mathbb{R}$ of $\mathpzc{D}$ is given by
\begin{equation}\label{eq:behD}
    \beh{\mathpzc{D}}(u)=\theta (\Delta^*(d_0,u)),
\end{equation}
for every $u\in X^*$.

From the transition function $\Delta^* :D\times X^*\to D$ we can extract a family of functions $\{\Delta_u\}_{u\in X^*}$, where $\Delta_u:D\to D$ is defined by $\Delta_u(a)=\Delta^*(a,u)$, for all $u\in X^*$ and $a\in D$. These functions will be also called \textit{transition functions}. If $u=x_1x_2\ldots x_k$, for $x_1,x_2,\ldots ,x_k\in X$, then 
\[
\Delta_u(a) =\Delta_{x_k}(\ldots (\Delta_{x_2}(\Delta_{x_1}(a)))),
\]
for every $a\in D$, that is, $\Delta_u$ is the composition $\Delta_u=\Delta_{x_1}\Delta_{x_2}\cdots \Delta_{x_k}$ of transition functions $\Delta_{x_1},\Delta_{x_2},\ldots ,\Delta_{x_k}$.

\subsection{Weighted automata over a vector space}

Let $V$ be a vector space over the field $\mathbb{R}$ of real numbers.~The third model of weighted automata~is~a \textit{weighted automaton over a vector space} (\textit{with vector states}), defined as a tuple $\mathpzc{A}=(S,\sigma,\delta,\Theta)$, where $S\subseteq V$~is~a~non\-empty set of vectors, called the  \textit{set of vector states}, $\sigma\in S$ is a vector called the \textit{initial vector state}, $\delta :S\times X\to S$ is a deterministic \textit{transition function}, and $\Theta:S\to \mathbb{R}$ is a function called the \textit{terminal weights function}. Here we also allow $S$ to be infinite.~Furthermore, in some sources $S$ is taken to be the entire space $V$, but here we allow $S$ to be only a subset of the space $V$ to enable it to be finite.~If the set $S$ of vector states is finite, then $\mathpzc{A}$ is called a \textit{finite weighted automaton over a vector space}, and if $V$ is a finite-dimensional vector space of dimension $n$, then $\mathpzc{A}$ is also said to be a \textit{finite-dimensional weighted automaton over a vector space of dimension $n$}.~If the cardinality of the set of states of $\mathpzc{A}$ is less than or equal to the cardinality of the set of states of any other weighted automaton over the vector space $V$, that we say that $\mathpzc{A}$ is a \textit{minimal weighted automaton over the vector space} $V$. 

As with crisp-deterministic weighted automata,  $\delta $ can be extended to a function $\delta^*:S\times X^*\to S$~by $\delta^*(\alpha,\varepsilon)=\alpha$ and $\delta^*(\alpha,ux)=\delta(\delta^*(\alpha,u),x)$, for all $\alpha \in S
$, $u\in X^*, x\in X$, and the function $\delta^*$ determines a family of functions $\{\delta_u\}_{u\in X^*}$, where $\delta_u:S\to S$~is defined by $\delta_u(\alpha)=\delta^*(\alpha,u)$, for all $u\in X^*$ and $\alpha\in S$. These functions are also called \textit{transition functions}. If $u=x_1x_2\ldots x_k$, for $x_1,x_2,\ldots ,x_k\in X$, then 
\[
\delta_u(\alpha ) =\delta_{x_k}(\ldots (\delta_{x_2}(\delta_{x_1}(\alpha )))),
\]
for every $\alpha \in V$, that is, $\delta_u$ is the composition $\delta_u=\delta_{x_1}\delta_{x_2}\cdots \delta_{x_k}$ of transition functions $\delta_{x_1},\delta_{x_2},\ldots ,\delta_{x_k}$. Then $\mathpzc{A}$ can be equivalently represented as a tuple $\mathpzc{A}=(S,\sigma,\{\delta_x\}_{x\in X},\Theta )$.~In addition, for every $u\in X^*$, with $\sigma_u$ we denote a vector from $S$ given by $\sigma_u=\delta^*(\sigma,u)=\delta_u(\sigma)$.  

The behavior $\beh{\mathpzc{A}}:X^*\to \mathbb{R}$ of the weighted automaton $\mathpzc{A}$ over a vector space $V$ is given by
\begin{equation}\label{eq:behAvs}
    \beh{\mathpzc{A}}(u)=\Theta (\delta^*(\sigma,u)),
\end{equation}
for every $u\in X^*$.~On the other hand, the function $\beh{\mathpzc{A}}_g$ generated by $\mathpzc{A}$ is defined with
\begin{equation}\label{eq:genAvs}
    \beh{\mathpzc{A}}_g (u)=\Vert \delta^*(\sigma,u)\Vert=\Vert \delta_u(\sigma)\Vert,
\end{equation}
for each $u\in X^*$, where $\Vert\cdot\Vert$ denotes some norm on the vector space $V$. If $V$ is a finite-dimensional space we will assume that $\Vert\cdot\Vert$ is the maximum norm $\Vert\cdot\Vert_{\infty}$ (see \eqref{eq:max.norm}), and if $V=\mathbb{R}^T$ is some function space (later we will consider the function space $\mathbb{R}^{X^*}$), then we will assume that $\Vert\cdot\Vert$ is the \textit{supremum norm} $\Vert\cdot\Vert_{\infty}$ (also called the \textit{uniform norm}) that is given by
\[
\Vert f\Vert_{\infty}=\sup_{t\in T} |f(t)|,
\]
for every $f:T\to \mathbb{R}$ (clearly, for $T=[1..n]$ we obtain \eqref{eq:max.norm}, i.e., the maximum norm).

\section{Relationships between different types of weighted automata}

Let $\mathpzc{A}$ and $\mathpzc{B}$ be weighted automata of any of the three types discussed in the previous section, where they can be of different types. If $\beh{\mathpzc{A}}=\beh{\mathpzc{B}}$, then $\mathpzc{A}$ and $\mathpzc{B}$ are said to be \textit{language-equivalent}.~On the other hand, if each of these automata is a weighted finite automaton or a weighted automaton over a vector space, where they do not have to be of the same type, and if $\beh{\mathpzc{A}}=\beh{\mathpzc{B}}$ and $\beh{\mathpzc{A}}_g=\beh{\mathpzc{B}}_g$, then $\mathpzc{A}$ and $\mathpzc{B}$ are said to be \textit{completely language-equivalent}.

\begin{theorem}
Let $\mathpzc{D} = (D, d_0, \Delta, \theta )$ be a crisp-deterministic weighted automaton over\, $\mathbb{R}$, let $V$ be a vector space and let $S\subseteq V$ be a set of vectors which has the same cardinality as $D$.

Then $\mathpzc{D}$ can be turned into a language-equivalent weighted automaton over the vector space $V$ having $S$ as its set of states.   
\end{theorem}

\begin{proof}
Let $\phi : S\to D$ be an arbitrary bijective function between $S$ and $D$. Then we can define an initial vector state $\sigma\in S$ by 
\[
\sigma=\phi^{-1}(d_0).
\] 
Let $\{\delta_x\}_{x\in X}$ be a family of transition functions defined in the following way: For each $x\in X$, 
\[
\delta_x(\alpha)=\phi^{-1}(\Delta(\phi(\alpha),x)),\quad\text{for every } \alpha\in S.  
\]
Let the terminal weights function $\Theta: S\to \mathbb{R}$ be defined by 
\[
\Theta(\alpha)=\theta(\phi(\alpha)),\quad\text{for every } \alpha\in S.  
\] 
 Then $\mathpzc{A}=(S,\sigma, \{\delta_x\}_{x\in X}, \Theta)$ is a weighted automaton over the vector space $V$ having $S$ as its set~of~states. 

Clearly, for each $\alpha\in S$ and $u\in X^*$ we have 
\[
\delta_u(\alpha)=\phi^{-1}(\Delta(\phi(\alpha),u).
\]
Furthermore, 
for every $u\in X^*$ the following holds
\begin{align*}
\beh{\mathpzc{A}}(u)&= \Theta(\delta^*(\sigma, u))=\Theta(\delta_u(\sigma))=
\theta (\phi(\delta_u(\sigma)))=\theta (\phi(\phi^{-1}(\Delta(\phi(\sigma),u))))=\\
&=
\theta (\phi(\phi^{-1}(\Delta(\phi(\phi^{-1} (d_0)),u))))=\theta (\Delta(d_0,u))= \beh{\mathpzc{D}}(u)  
\end{align*}
and therefore, $\mathpzc{A}$ and $\mathpzc{D}$ are  language-equivalent.
\end{proof}

T. Li, G. Rabusseau and D. Precup in \cite{LRP.18} defined a \textit{nonlinear weighted finite automaton} over the field of real numbers $\mathbb R$ as a tuple $(\sigma,\{\delta_x\}_{x\in X},\Theta )$, where $\sigma\in \mathbb{R}^n$ is a vector of initial weights, $\{\delta_x\}_{x\in X}$ is a family of transformations such that $\delta_x:\mathbb{R}^n\to \mathbb{R}^n$, for each $x\in X$, which are called transition functions, and $\theta :\mathbb{R}^n\to \mathbb{R}$ is a termination function.~This definition is almost identical to our definition of a weighted automaton over a vector space.~The difference is that the set $S$ of vector states is not explicitly stated in the mentioned paper, but we can assume that $S$ is the smallest set of vectors from $\mathbb{R}^n$ that contains $\sigma $ and is closed for all transformations from the family $\{\delta_x\}_{x\in X}$, that is, $S=\{\sigma_u\}_{u\in X^*}$, where $\sigma_u=\delta^*(\sigma,u)=\delta_u(\sigma)$, for each $u\in X^*$.~Another difference is that the transformations $\{\delta_x\}_{x\in X}$ are defined on $\mathbb{R}^n$, but nothing changes significantly if we replace them with their restrictions on $S$.~The third difference is that Lee, Rabusseau and Precup considered automata over the finite-dimensional vector space $\mathbb{R}^n$, while in our definition we provide the possibility that the underlying vector space $V$ be also infinite-dimensional. 
Therefore, the concepts of a nonlinear weighted finite automaton and a weighted automaton over a vector space are almost the same.~Let us note that the adjective "finite" in the name of nonlinear weighted finite automata does not refer to the finiteness of the set $S$ of vector states, as in our definition, but to the finite dimension of the space $\mathbb{R}^n$.~In addition, regardless of the adjective nonlinear in the name, in the nonlinear weighted finite automaton among the transformations $\delta_x$, $x\in X$, there can be both linear and nonlinear ones. 

Here, a weighted automaton $\mathpzc{A}=(S,\sigma,\{\delta_x\}_{x\in X},\Theta )$ over a vector space $V$ is defined to be \textit{linear} if for any $x\in X$ the function $\delta_x$ is the restriction of some linear operator $\delta_{\!\!x}':\mathrm{span}(S)\to \mathrm{span}(S)$~to~the~set $S$, and also, $\Theta $ is the restriction of some linear functional (linear form)  $\Theta':\mathrm{span}(S)\to \mathbb{R}$ to the set~$S$,~i.e.,%
\begin{equation}\label{eq:lin.waovs}
\delta_{\!\!x}'(s\alpha+t\beta)=s\delta_{\!\!x}'(\alpha)+t\delta_{\!\!x}'(\beta ), \qquad 
    \Theta' (s\alpha+t\beta)=s\Theta'(\alpha)+t\Theta'(\beta ),    
\end{equation}
for all $x\in X$, $\alpha,\beta \in \mathrm{span}(S)$ and $s,t\in \mathbb{R}$.~Otherwise, if some of the mappings $\delta_x$, $x\in X$, and $\Theta $ can not be~represented in this way, then $\mathpzc{A}$ is said to be a \textit{nonlinear weighted automaton over a vector space}.~If $V\subseteq \mathbb{R}^n$, for some $n\in \mathbb N$, then $\mathpzc{A}$ is linear if and only if for each $x\in X$ there is a matrix $M_x\in \mathbb{R}^{n\times n}$ such that $\delta_x(\alpha )=\alpha\cdot M_x$, for each $\alpha \in S$, and there is also a vector $\tau\in \mathbb{R}^n$ such that $\Theta(\alpha)=\alpha\cdot \tau $, for each $\alpha\in S$ (here $\alpha\cdot \tau $ denotes the scalar product of $\alpha $~and~$\tau$). 

Let us note that our linear weighted autom\-ata over a vector space are almost identical to autom\-ata studied by Balle, Gourdeau and Panangaden in \cite{BGP.22} (which are called there only weighted finite~automata), the only difference is that there $V$ was assumed to be a finite-dimensional space and~$S=V$. However, even this small difference concerning the set of vector states can be significant.~Namely,~let $\mathpzc{A}=(S,\sigma,\{\delta_x\}_{x\in X},\Theta )$ be any linear 
weighted automaton over the vector space $V=\mathbb{R}^n$ with the set of vector states $S\subset V$ such that $\mathrm{span}(S)\ne V$, for each $x\in X$ let $\delta_x$ be the restriction of some linear operator $\delta_{\!\!x}':\mathrm{span}(S)\to \mathrm{span}(S)$ to the set $S$, and let $\Theta $ be the restriction of some linear functional   $\Theta':\mathrm{span}(S)\to \mathbb{R}$ to the set~$S$. Then we can~easily~extend any $\delta_{\!\!x}'$ to an operator $\delta_{\!\!x}'':V\to V$ which is not linear, for example, by taking  $\delta_{\!\!x}''$ to concide with $\delta_{\!\!x}'$ on $\mathrm{span}(S)$ and  
\[
\delta_{\!\!x}''\bigl(\begin{bmatrix} s_1 & s_2 & \ldots & s_n\end{bmatrix}\bigr)=\begin{bmatrix} s_1^2 & s_2^2 & \ldots & s_n^2\end{bmatrix},
\]
for every $\begin{bmatrix} s_1 & s_2 & \ldots & s_n\end{bmatrix}\in V\setminus \mathrm{span}(S)$, and we can extend $\Theta'$ to a non-linear function $\Theta'':V\to \mathbb{R}$. Therefore, $\mathpzc{A}''=(V,\sigma,\{\delta_{\!\!x}''\}_{x\in X},\Theta'')$ is a non-linear weighted automaton over the vector space $V$ with the set of vector states $V$, but if we assume that the set of vector states is $S$, then $\mathpzc{A}''$ becomes linear.

The following theorem explains the connection between finite-dimensional linear weighted automata over a vector space and weighted finite automata.

\begin{theorem}\label{th:lwavs.wfa}
    Every finite-dimensional linear weighted automaton over a vector space can be turned into a completely language-equivalent weighted finite automaton.
    
    Conversely, every weighted finite automaton can be turned into a completely language-equivalent finite-dimensional linear weighted automaton over a vector space.
\end{theorem}

\begin{proof}
    Let  $\mathpzc{A}=(S,\sigma,\delta,\Theta)$ be a finite-dimensional linear weighted automaton over a vector space $V$ of dimension $n$.~Since $\mathpzc{A}$ is linear, we have that for each $x\in X$ there exists a matrix $M_x\in \mathbb{R}^{n\times n}$ such that $\sigma\cdot M_x=\delta(\sigma, x)$, and moreover, there exists a vector $\tau \in \mathbb R^n$ such that $\Theta (\sigma)=\sigma\cdot \tau$. In this way, we obtain a weighted finite automaton $\mathpzc{A}'$ which is given by the linear representation  $\mathpzc{A}'=(n, \sigma, \{M_x\}_{x\in X}, \tau)$. 
    
    Now, for every $u\in X^+$ such that $u=x_1x_2\cdots x_k$, where $x_1,x_2,\dots, x_k\in X$, we have 
    \begin{align*}
        \beh{\mathpzc{A}'}(u)=\sigma\cdot M_{x_1}\cdot M_{x_2}\cdot\ \cdots\ \cdot M_{x_k}\cdot \tau = \sigma\cdot M_u\cdot \tau=\delta^*(\sigma, u)\cdot \tau=\Theta (\delta^*(\sigma, u))=\beh{\mathpzc{A}}(u),
    \end{align*}
and
\begin{align*}
        \beh{\mathpzc{A}'}(\varepsilon)=\sigma\cdot\tau = \delta^*(\sigma, \varepsilon)\cdot \tau=\Theta (\delta^*(\sigma, \varepsilon))=\beh{\mathpzc{A}}(\varepsilon).
    \end{align*}
 Finally, we have that
 \[
 \beh{\mathpzc{A}'}_g(u)=\Vert \sigma_u\Vert_{\infty} = \Vert \sigma\cdot M_u\Vert_{\infty}=\Vert\delta^*(\sigma, u)\Vert_{\infty}=\beh{\mathpzc{A}}_g(u),
 \]
for every $u\in X^*$. 

Therefore, we have proved that the finite-dimensional linear weighted automaton $\mathpzc{A}$ is completely language-equivalent to the weighted finite automaton $\mathpzc{A}'$.

Conversely, let $\mathpzc{A}=(A,\sigma,\delta,\tau)$ be a weighted finite automaton with $n$ states.~We define a weighted automaton $\mathpzc{A}_N=(S_N,\sigma^N,\delta^N,\Theta^N)$ over the vector space $\mathbb{R}^n$  in the following way: we set
\[
{S}_N=\{\sigma_u\, |\, u\in X^* \}, \quad \sigma^N=\sigma ,
\]
and we define functions $\delta^N:S_N\times X\to S_N$ and $\Theta^N:S_N\to \mathbb R$ by
\[
{\delta}^N(\sigma_u,x)=\sigma_u\cdot \delta_x=\sigma_{ux},\qquad {\Theta}^N(\sigma_u)=\sigma_u \cdot \tau,
\]
for every $u\in X^*$.~The automaton $\mathpzc{A}_N$ is obviously well-defined and is called in \cite{ICBP.10} the \textit{Nerode automaton} of the automaton $\mathpzc{A}$.
It remains to prove that the Nerode automaton $\mathpzc{A}_N$ is completely language-equivalent to the original weighted finite automaton $\mathpzc{A}$.

For an arbitrary word $u\in X^*$ we have that
\begin{align*}
\beh{\mathpzc{A}_N}(u)= \Theta^N((\delta^N)^*(\sigma, u))=\Theta^N (\sigma_u)=\sigma_u\cdot \tau = \beh{\mathpzc{A}}(u).
    \end{align*}
and also, 
\begin{align*}
        \beh{\mathpzc{A}_N}_g(u)=
        \Vert (\delta^N)^*(\sigma, u)\Vert_{\infty}
    =\Vert\sigma_u\Vert_{\infty}=\beh{\mathpzc{A}}_g(u).
    \end{align*}
    Therefore, the Nerode automaton $\mathpzc{A}_N$ of $\mathpzc{A}$ is completely language-equivalent to $\mathpzc{A}$.
\end{proof}

Let $\mathpzc{A}=(S,\sigma,\{\delta_x\}_{x\in X},\Theta )$ be a finite weighted automaton over the vector space $V=\mathbb{R}^n$, and let us assume that $S=\{\alpha_1,\alpha_2,\ldots ,\alpha_m\}$ and $X=\{x_1,x_2,\ldots ,x_s\}$. Let us form $m\times n$-matrices $N$ and $N_{x_i}$, $i\in [1..s]$,  and a column vector ($1\times m$-matrix) $\vartheta$ such that rows in $N$ are $\alpha_1, \alpha_2, \ldots ,\alpha_m$, rows in $N_{x_i}$ are $\delta_{x_i}(\alpha_1), \delta_{x_i}(\alpha_2), \ldots $, $\delta_{x_i}(\alpha_m)$, and entries in $\vartheta$ are $\Theta(\alpha_1), \Theta(\alpha_2), \ldots ,\Theta(\alpha_m)$, in that order, i.e.
\[
N=\begin{bmatrix} \alpha_1 \\ \alpha_2 \\ \vdots \\ \alpha_m \end{bmatrix}, \quad
N_{x_i}= \begin{bmatrix}\delta_{x_i}(\alpha_1)\\ \delta_{x_i}(\alpha_2)\\ \vdots \\ \delta_{x_i}(\alpha_m) \end{bmatrix}, \ \ i\in [1..s], \quad 
\vartheta= \begin{bmatrix}\Theta(\alpha_1)\\ \Theta(\alpha_2)\\ \vdots \\ \Theta(\alpha_m) \end{bmatrix}
\]
We will call $N$ the \textit{state matrix} and $\vartheta $ the \textit{terminal vector}, while for each $i\in [1..s]$, we call $N_{x_i}$~the~\textit{destination matrix} corresponding to the input letter $x_i$.

The following theorem provides a procedure for testing the linearity of a finite weighted automaton over a finite-dimensional vector space.

\begin{theorem}\label{th:lin}
Let $\mathpzc{A}=(S,\sigma,\{\delta_x\}_{x\in X},\Theta )$ be a finite weighted automaton over the vector space $V\subseteq\mathbb{R}^n$. Then $\mathpzc{A}$ is linear if and only if the matrix $N$ has the same rank as the augmented matrix 
\[
[\,N\mid N_{x_1}\mid N_{x_2}\mid\ldots \mid N_{x_s}\mid \vartheta \,].
\]
\end{theorem}

\begin{proof}
First we prove that $\mathpzc{A}$ is linear if and only if each of the following equations is solvable
\begin{equation}\label{eq:matr.eq}
N\cdot X_1=N_{x_1}, \ \ N\cdot X_2=N_{x_2}, \ \ \ldots ,\ \ N\cdot X_s=N_{x_s}, \ \ 
N\cdot \chi = \vartheta ,
\end{equation}
where $X_1,X_2,\ldots ,X_n$ are unknown $n\times n$-matrices, and $\chi $ is an unknown vector taking values in $\mathbb{R}^n$.

Assume that $\mathpzc{A}$ is linear, i.e., that there are matrices $M_{x_i}\in \mathbb{R}^{n\times n}$, for each $i\in [1..s]$, and a vector $\tau\in \mathbb{R}^n$ such that 
\begin{equation}\label{eq:linearity}
    \delta_{x_i}(\alpha_j)=\alpha_j\cdot M_{x_i},\ \ \text{for all}\ i\in [1..n],\ j\in [1..m],\qquad \Theta(\alpha_j)=\alpha_j\cdot \tau ,\ \ \text{for each}\ j\in [1..m].
\end{equation} 
According to the well-known row rule for matrix multiplication we obtain that
\begin{equation}\label{eq:matr.eq.sol}
N\cdot M_{x_i}=
\begin{bmatrix}
    \alpha_1\cdot M_{x_i} \\ \alpha_2\cdot M_{x_i} \\ \cdots \\ \alpha_m\cdot M_{x_i}
\end{bmatrix}=
\begin{bmatrix}
    \delta_{x_i}(\alpha_1) \\ \delta_{x_i}(\alpha_2) \\ \cdots \\ \delta_{x_i}(\alpha_m)
\end{bmatrix}=N_{x_i},\qquad
N\cdot \tau=
\begin{bmatrix}
    \alpha_1\cdot \tau \\ \alpha_2\cdot \tau \\ \cdots \\ \alpha_m\cdot \tau
\end{bmatrix}=
\begin{bmatrix}
    \Theta(\alpha_1) \\ \Theta(\alpha_2) \\ \cdots \\ \Theta(\alpha_m)
\end{bmatrix}=\vartheta ,    
\end{equation}
for every $i\in [1..s]$, and we conclude that $X_1=M_{x_1}$, $X_2=M_{x_2}$, \ldots , $X_s=M_{x_s}$ and $\chi=\tau$ are solutions of equations from \eqref{eq:matr.eq}.

Conversely, let any of the equations from \eqref{eq:matr.eq} is solvable, and assume that $X_i=M_{x_i}$, for $i\in [1..s]$, and $\chi=\tau $, are solutions to these equations.~From there we obtain that \eqref{eq:matr.eq.sol} holds, which further implies that \eqref{eq:linearity} holds, and therefore, $\mathpzc{A}$ is a linear weighted automaton over a vector space $V$.

Next we prove that each of the equations from \eqref{eq:matr.eq} is solvable if and only if the rank of $N$ is equal to the rank of the augmented matrix $N^*=[\,N\mid N_{x_1}\mid N_{x_2}\mid\ldots \mid N_{x_s}\mid \vartheta \,]$.~Let $R$ be the reduced row echelon form of the augmented matrix $N^*$, and let $R$ be partitioned as follows:
\[
R=[\,R_0\,|\,R_1\,|\,R_2\,|\,\ldots\,|\,R_s\,|\,\nu\,]
\]
where $R_0,R_1,R_2,\ldots ,R_s$ are $m\times n$-matrices and $\nu $ is a column vector of dimension $m$.~Then $R_0$ is the reduced row echelon form of $N$, for each $i\in [1..s]$, $[\,R_0\,|\,R_i\,]$ is the reduced row echelon form of $[\,N\,|\,N_{x_i}\,]$,  and $[\,R_0\,|\,\nu\,]$ is the reduced row echelon form of $[\,N\,|\,\vartheta\,]$.

For an arbitrary $i\in [1..s]$ we have that the equation $N\cdot X_i=N_{x_i}$ is solvable if and only the equation $N\cdot c_k({X_i})=c_k({N_{x_i}})$ is solvable for every $k\in [1..n]$.~On the other hand, for every $k\in [1..n]$ we have that 
$N\cdot c_k({X_i})=c_k({N_{x_i}})$ is solvable if and only if $\mathrm{rank}(N)=\mathrm{rank}([\,N\,|\,c_k({N_{x_i}})])$, and this holds if and only if for every $j\in [1..m]$ for which the $j$th row of $R_0$ is the zero vector it follows that the $j$th coordinate of $c_k({N_{x_i}})$ is equal to zero. Similarly, the equation $N\cdot \chi =\vartheta$ is solvable if and only if for every $j\in [1..m]$ for which the $j$th row of $R_0$ is the zero vector it follows that the $j$th coordinate of $\vartheta$ is equal to zero.

Therefore, we conclude that all equations in \eqref{eq:matr.eq} are solvable if and only if for every $j\in [1..m]$ for which the $j$th row of $R_0$ is the zero vector, we have that all remaining entries in the $j$th row of the matrix $R=[\,R_0\,|\,R_1\,|\,R_2\,|\,\ldots\,|\,R_s\,|\,\nu\,]$ are equal to zero, and this is equivalent to $\mathrm{rank}(N)=\mathrm{rank}(N^*)$. 

Let us note that if $R_0$ does not have zero rows then all equations in \eqref{eq:matr.eq} are solvable if and only if
\[
\mathrm{rank}(N)=\mathrm{rank}(N^*)=\mathrm{rank}([\,N\,|\,N_{x_i}\,])=
\mathrm{rank}([\,N\,|\,\vartheta\,])=
\mathrm{rank}([\,N\,|\,c_k(N_{x_i})\,])=m\ (\leqslant n),
\]
for all $i\in [1..s]$ and $k\in [1..n]$.
This completes the proof of the theorem.
\end{proof}

\begin{example}\label{ex:non.lin}\rm 
Let $\mathpzc{A}=(S,\sigma,\{\delta_x\}_{x\in X},\Theta )$ be a finite weighted automaton with three vector states over the vector space $\mathbb{R}^2$ and an alphabet $X=\{x,y\}$, where the vector states $\sigma=\alpha_1$, $\alpha_2$ and $\alpha_3$, and the terminal vector $\vartheta $ are given with
\[
\sigma=\alpha_1=\begin{bmatrix} 1 & 0 \end{bmatrix}, \quad 
\alpha_2=\begin{bmatrix} 0 & 1 \end{bmatrix}, \quad
\alpha_3=\begin{bmatrix} 1 & 1 \end{bmatrix}, \quad 
\vartheta=\begin{bmatrix} 0 \\ 0 \\ 1 \end{bmatrix},
\]
while the transition graph is the one given in Figure \ref{fig:tr.gr}.

\begin{figure}[h]
\begin{center}
\begin{tikzpicture}[
roundnode/.style={circle, draw=black!50, fill=black!5, thick, minimum size=5mm}
]
\node[roundnode]        (a1)                                {$\alpha_1$};
\node[roundnode]        (a2)       [right=32mm of a1]       {$\alpha_2$};
\node[roundnode]        (a3)       [above right=20mm of a1]      {$\alpha_3$};

\draw[->,thick,draw=black!50] (a1.east) -- (a2.west) node[pos=0.5,below] {$x$};
\draw[->,thick,draw=black!50] (a1.north east) -- (a3.south west) node[pos=0.5,above,sloped] {$y$};
\draw[->,thick,draw=black!50] (a3.south east) -- (a2.north west) node[pos=0.5,above,sloped] {$y$};
 \tikzset{every loop/.style={->,min distance=8mm,in=-45,out=45,looseness=5}}
 \path (a2) edge [loop right,thick,draw=black!50] node {$x,y$} (a2);
 \tikzset{every loop/.style={->,min distance=8mm,in=45,out=135,looseness=5}}
 \path (a3) edge [loop above,thick,draw=black!50] node {$x$} (a3);
\end{tikzpicture}
\end{center}
\vspace{-10mm}
\caption{The transition graph of the automaton $\mathpzc{A}$ from Example \ref{ex:non.lin}}
\label{fig:tr.gr}
\end{figure}
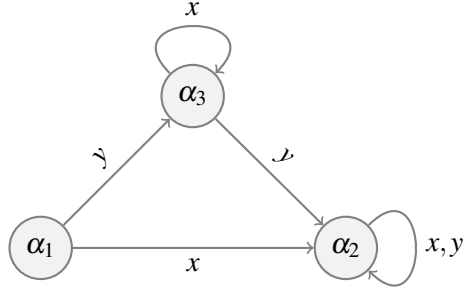

Then the augmented matrix $N^*=[N\mid N_x\mid N_y\mid \vartheta]$ and its reduced row echelon form $RREF(N^*)$ are given by 
\[
N^*=\left[
\begin{array}{cc|cc|cc|c}
   1 & 0 & 0 & 1 & 1 & 1 & 0 \\
   0 & 1 & 0 & 1 & 0 & 1 & 0 \\
   1 & 1 & 1 & 1 & 0 & 1 & 1 
\end{array}\right], \quad
RREF(N^*)=\left[
\begin{array}{cc|cc|cc|c}
   1 & 0 & 0 & 1 & 1 & 1 & 0 \\
   0 & 1 & 0 & 1 & 0 & 1 & 0 \\
   0 & 0 & 1 & -1 & -1 & -1 & 1 
\end{array}\right],
\]
and it is clear that $\mathrm{rank}(N)=2\ne 3=\mathrm{rank}(N^*)$. From there we get that $\mathpzc{A}$ is nonlinear.

Let us also note that the equations \eqref{eq:matr.eq} become
\[
\begin{bmatrix} 1 & 0 \\ 0 & 1 \\ 1 & 1 \end{bmatrix}\cdot 
\begin{bmatrix} x_{11} & x_{12} \\ x_{21} & x_{22} \end{bmatrix} =
\begin{bmatrix} 0 & 1 \\ 0 & 1 \\ 1 & 1 \end{bmatrix}, \qquad
\begin{bmatrix} 1 & 0 \\ 0 & 1 \\ 1 & 1 \end{bmatrix}\cdot 
\begin{bmatrix} y_{11} & y_{12} \\ y_{21} & y_{22} \end{bmatrix}=
\begin{bmatrix} 1 & 1 \\ 0 & 1 \\ 0 & 1 \end{bmatrix}, \qquad
\begin{bmatrix} 1 & 0 \\ 0 & 1 \\ 1 & 1 \end{bmatrix}\cdot 
\begin{bmatrix} x \\ y \end{bmatrix}=
\begin{bmatrix} 0 \\ 0 \\ 1 \end{bmatrix}.
\]
It can be seen that none of these equations has a solution.
\end{example}

For a word function $f:X^*\to \mathbb{R}$ and a word $u\in X^*$, a word function $f_u:X^*\to \mathbb{R}$ defined by
\begin{equation}\label{eq:derivative}
    f_u(v)=f(uv),\qquad\text{for every $v\in X^*$},
\end{equation}
is called the \textit{derivative} of $f$ with respect to $u$. Derivatives of conventional languages are also known as \textit{right quotients}, \textit{quotients} or \textit{residuals} of languages. 

We define a weighted automaton $\mathpzc{A}_f=(S_f,\sigma^f,\delta^f,\Theta^f)$
over a vector space $\mathbb{R}^{X^*}$ as follows: the set $S_f$ of vector states is given by $S_f=\{f_u\,|\, u\in X^*\}$, $\sigma^f=f$, and functions 
$\delta^f:S^f\times X\to S^f$ and $\Theta^f:S^f\to \mathbb{R}$ are given by
\begin{equation}\label{eq:der.aut}
    \delta^f(g,x)=g_x, \quad \Theta^f(g)=g(\varepsilon), \qquad \text{for all $g\in S_f$ and $x\in X$.}
\end{equation}
It is clear that $\mathpzc{A}_f$ is well-defined, and we will call it the \textit{derivative automaton} of the word function $f$.

For a word function $f:X^*\to \mathbb{R}$, the \textit{prefix closure} of $f$ is a word function $\overbar{f}:X^*\to \mathbb{R}$ defined by
\begin{equation}\label{eq:pref.clos}
    \overbar{f}(u)=\sup_{v\in X^*}|f(uv)|=\Vert f_u\Vert_{\infty},
\end{equation}
for every $u\in X^*$.

\begin{theorem}\label{th:der.aut}
The derivative automaton $\mathpzc{A}_f$ of a word function $f\in \mathbb{R}^{X^*}$ is a linear weighted automaton over the  vector space \,$\mathbb{R}^{X^*}$ that computes $f$ and generates the prefix closure $\overbar{f}$ of $f$.

In addition, $\mathpzc{A}_f$ is a minimal weighted automaton over a vector space which computes $f$. 
\end{theorem}

\begin{proof} For the sake of simplicity, let us assume that $\mathbb{R}^{X^*}=V$.

First, we prove that $\mathpzc{A}_f$ is linear.~To that end, define functions $\delta_x:V\to V$, for every $x\in X$, and $\Theta :V\to \mathbb{R}$ as follows: $\delta_x(g)=g_x$ and $\Theta(g)=g(\varepsilon)$, for all $g\in V$ and $x\in X$.~Further, consider arbitrary $g,h\in V$ and $s,t\in \mathbb{R}$. Then for arbitrary $x\in X$ and $u\in X^*$ we have that
\[
(sg+th)_x(u)=(sg+th)(xu)=sg(xu)+th(xu)=sg_x(u)+th_x(u)=(sg_x+th_x)(u),
\]
so we conclude that $(sg+th)_x=sg_x+th_x$, and now we have that
\begin{align*}
  \delta_x(sg+th)&=(sg+th)_x=sg_x+th_x=s\delta_x(g)+t\delta_x(h), \\ 
  \Theta (sg+th)&= (sg+th)(\varepsilon)= sg(\varepsilon)+th(\varepsilon) =s\Theta(g)+t\Theta(h).  
\end{align*}
Hence, $\delta_x$, $x\in X$, and $\Theta$ are linear operators on $V$, and it is clear that $\delta^f_x$ (where $\delta^f_x(g)=\delta^f(g,x)$) is the  restriction of $\delta_x$ to $S_f$, for each $x\in X$, and $\Theta^f$ is 
the restriction of $\Theta$ on $S_f$. This means that $\mathpzc{A}_f$ is a linear weighted automaton over the vector space $V$. 
Next, for an arbitrary $u\in X^*$ we have that
\begin{align*}
  &\beh{\mathpzc{A}_f}(u)= \Theta((\delta^f)^*(\sigma^f,u))= \Theta((\delta^f)^*(f,u)) = \Theta(f_u)=f_u(\varepsilon)=f(u\varepsilon)=f(u), \\
  &\beh{\mathpzc{A}_f}_g(u)=\Vert (\delta^f)^*(\sigma^f,u)\Vert_{\infty}=
  \Vert (\delta^f)^*(f,u)\Vert_{\infty}=\Vert f_u\Vert_{\infty}=\overbar{f}(u),
\end{align*}
which means that the automaton $\mathpzc{A}_f$ computes $f$ and generates $\overbar{f}$.

Let $\mathpzc{A}=(S,\sigma,\delta,\Theta)$ be an arbitrary weighted automaton over a vector space that computes $f$, and let $S'=\{\sigma_u\,|\, u\in X^*\}\subseteq S$. Define a function $\phi:S'\to S_f$ by putting that $\phi(\sigma_u)=f_u$, for every $u\in X^*$. First we prove that $\phi $ is well-defined, i.e., that for any $u,v\in X^*$, from $\sigma_u=\sigma_v$ it follows $f_u=f_v$. Thus, consider $u,v\in X^*$ such that $\sigma_u=\sigma_v$, and an arbitrary $w\in X^*$. Then
\begin{align*}
    f_u(w)&=f(uw)=\beh{\mathpzc{A}}(uw)=\Theta (\delta^*(\sigma,uw))=\Theta (\delta^*(\delta^*(\sigma,u),w)) \\
    &=\Theta (\delta^*(\sigma_u,w))=\Theta (\delta^*(\sigma_v,w))
    =\Theta (\delta^*(\delta^*(\sigma,v),w))\\&=\Theta (\delta^*(\sigma,vw))= \beh{\mathpzc{A}}(vw)=f(vw)=f_v(w),
\end{align*}
so we conclude that $f_u=f_v$.~Therefore, we get that $\phi $ is a well-defined function, and it is obvious that $\phi$ is surjective. This means that the cardinality of $S_f$ is less than or equal to the cardinality of $S'$, which is less than or equal to the cardinality of $S$.~From there, we conclude that $\mathpzc{A}_f$ is a minimal weighted automaton over a vector space which computes $f$.
\end{proof}

\bibliographystyle{eptcs}
\bibliography{approxbib}

\end{document}